\documentclass[11pt,fleqn]{article}
\usepackage[utf8]{inputenc}
\usepackage{amsmath,amssymb,graphicx,amsthm,dsfont}
\usepackage{color,soul}
\usepackage{xspace}
\usepackage{enumerate}
\usepackage{mathtools}
\usepackage{bbm}
\usepackage[numbers]{natbib}
\usepackage{fullpage}

\usepackage{multicol}
\usepackage{booktabs}
\usepackage{paralist}

\usepackage{subcaption}
\usepackage[noend,ruled,linesnumbered]{algorithm2e} 
\usepackage{xifthen}

\usepackage[linecolor=green!70!black, backgroundcolor=green!10, bordercolor=black, textsize=scriptsize]{todonotes}

\usepackage{mdframed}

\usepackage[colorlinks,citecolor=blue!80!black,linkcolor=red!60!black,pagebackref]{hyperref}
\renewcommand*{\backref}[1]{}
\renewcommand*{\backrefalt}[4]{%
\ifcase #1%
\marginpar{\tiny no cite}
\or
 $\rightarrow$~p.~#2.%
\else
  $\rightarrow$~pp.~#2.%
\fi
}

\usepackage{tikz}
\usetikzlibrary{decorations,arrows,petri,topaths,backgrounds,shapes,positioning,fit,calc,decorations.pathreplacing,patterns,intersections,decorations.pathmorphing,matrix,angles,quotes}

\tikzstyle{thickline} = [line width=1.8pt]
\tikzstyle{gl} = [draw, gray]
\tikzstyle{pl} = [draw, orange]
\tikzstyle{ml} = [draw, blue]
\tikzstyle{sett} = [draw, thick, purple]
\tikzstyle{ele} = [draw, thick, green!70!black]

\makeatletter
\newcommand{\gettikzxy}[3]{%
  \tikz@scan@one@point\pgfutil@firstofone#1\relax
  \edef#2{\the\pgf@x}%
  \edef#3{\the\pgf@y}%
}
\makeatother

\usepackage{cleveref}

\newtheorem{construction}{Construction}
\newtheorem{theorem}{Theorem}
\newtheorem{claim}{Claim}

\theoremstyle{definition}

\crefname{table}{Table}{Tables}
\crefname{figure}{Figure}{Figures}
\crefname{theorem}{Theorem}{Theorems}
\crefname{definition}{Definition}{Definitions}
\crefname{corollary}{Corollary}{Corollaries}
\crefname{observation}{Observation}{Observations}
\crefname{lemma}{Lemma}{Lemmas}
\crefname{example}{Example}{Examples}
\crefname{reduction}{Reduction}{Reductions}
\crefname{construction}{Construction}{Constructions}
\crefname{subsection}{Subsection}{Subsections}
\crefname{section}{Section}{Sections}
\crefname{proposition}{Proposition}{Propositions}
\crefname{algorithm}{Algorithm}{Algorithms}
\crefname{claim}{Claim}{Claims}

\tikzstyle{alter} = [draw, circle, minimum size=4ex, inner sep=1pt, text centered, align=center]

\tikzstyle{nn} = [draw, circle, inner sep=.7pt,fill=black]
\tikzstyle{dnode} = [nn, rectangle, blue, fill=blue]

\newcommand{\decprob}[3]{%
  \begin{center}%
    \begin{minipage}{0.9\linewidth}%
      \textsc{#1}\\
      \textbf{Input:} #2\\
      \textbf{Question:} #3
    \end{minipage}%
  \end{center}%
}

\newcommand{\Clique}{\textsc{Clique}}

\newcommand{\misr}{\ensuremath{\mathsf{\rho}}}
\newcommand{\repre}{\ensuremath{\omega}}
\newcommand{\rank}{\ensuremath{\mathsf{rk}}}
\newcommand{\RR}{\ensuremath{\mathcal{R}}}
\newcommand{\ppp}{\ensuremath{\mathcal{P}}}
\newcommand{\bound}{\ensuremath{\beta}}
\newcommand{\CC}{\textsc{CC}}
\newcommand{\Monroe}{\textsc{Monroe}}
\newcommand{\mw}{\textsc{Multi-Winner}}
\newcommand{\MW}{{Multi-Winner}}

\newcommand{\myemph}[1]{{\color{green!40!black}\emph{#1}}}

\newcommand{\seq}[1]{\ensuremath{\langle #1 \rangle}}

\newcommand{\enn}{\ensuremath{\hat{n}}}
\newcommand{\emm}{\ensuremath{\hat{m}}}
\newcommand{\kk}{\ensuremath{\mathsf{s}}}
\newcommand{\hide}[1]{}

\title{Parameterized Intractability for \MW{} Election under the~Chamberlin-Courant Rule and the Monroe Rule}

\author{Jiehua Chen \and Sanjukta Roy}
\date{TU Vienna, Austria\\
  \texttt{\{jiehua.chen, sanjukta.roy\}@tuwien.ac.at}}

\begin{document}

\maketitle

\begin{abstract}
  Answering an open question by Betzler et al. [Betzler et al., JAIR~'13], we resolve the parameterized complexity of the multi-winner determination problem under two famous representation voting rules: the \emph{Chamberlin-Courant} (in short CC) rule~[Chamberlin and Courant, APSR~'83] and the \emph{Monroe} rule~[Monroe, APSR~'95].
We show that under both rules, the problem is W[1]-hard with respect to the sum~$\bound$ of misrepresentations, thereby precluding the existence of any $f(\bound)\cdot |I|^{O(1)}$-time algorithm, where $|I|$ denotes the size of the input instance. 
\end{abstract}

\section{Introduction}

In the \CC-\mw{} problem, we are given a preference profile for a set~$V$ of voters and a set~$C$ of alternatives, and a number~$k$.
We aim to find an assignment~$\repre\colon V\to C$ of the voters to the alternatives such that $|\repre(V)|=k$ and the sum of misrepresentations of~$\repre$ is minimum.
Here, a \myemph{preference profile} %
is a tuple~$(C,V,\RR)$, where \myemph{$C$} denotes a set of $m$~alternatives with $C=\{c_1,\cdots, c_m\}$,
\myemph{$V$} denotes a set of $n$ voters with $V=\{v_1,\cdots, v_n\}$,
and \myemph{$\RR$} is a collection of $n$ preference orders~$\RR=(\succ_1,\cdots, \succ_n)$ such that
each voter~$v_i\in V$ has a preference order~$\succ_i$ over~$C$ in~$\RR$. 
A \myemph{preference order} over~$C$ is just a linear order on~$C$.
Given $C=\{1,2,3\}$, a possible preference order over~$C$ may be $2 \succ 3 \succ 1$, which
indicates that $2$ is preferred to~$3$ and $1$, and $3$ is preferred to~$1$.
For a preference order~$\succ_i$ of voter~$v_i$ and an arbitrary alternative~$c_j\in C$, the \myemph{rank} of $c_j$ in~$\succ_i$ is the number of alternatives that are preferred to~$c_j$ and we use \myemph{$\rank_{v_i}(c_j)$} to refer to this number, i.e., \myemph{$\rank_{v_i}(c_j)=|\{c_z \in C \mid c_z\succ_i c_j\}|$}.
The \myemph{sum of misrepresentations} of $\repre$ is defined as the sum of the ranks of the assigned alternatives, i.e., \myemph{$\misr(\ppp, \repre) = \sum_{v_i\in V}\rank_{v_i}(\repre(v_i))$}.

In the \Monroe-\mw{} problem, the input is the same as the one for \CC-\mw{}, but the task is to find a \emph{proportional} assignment~$\repre$ such that $|\repre(V)|=k$ and the sum of misrepresentations is minimum.
Assignment~$\repre$ is \myemph{proportional} if each assigned alternative represents roughly the same portion of the voters,
i.e., $\lfloor |V| / k \rfloor \le |\{v\in V\mid \repre^{-1}(c)\}|\le \lceil{|V|/k\rceil}$ holds for all~$c\in \repre(V)$.

Both \CC-\mw{} and \Monroe-\mw{} are NP-hard~\cite{lu2011budgeted} and contained in XP
with respect to the minimum sum~$\bound$ of misrepresentations of any valid assignment~\cite{BetzlerMW2013}
since they can be solved in $m^\bound\cdot (n+m)^{O(1)}$ time by guessing the alternatives (assigned to the voters) which each induces positive misrepresentations.
It was an open question whether the dependence of the exponent of the running time  of~$\bound$ can be removed to obtain a fixed-parameter algorithm for~$\bound$~\cite{BetzlerMW2013}.
We answer this open question negatively, by showing W[1]-hardness for~$\bound$.

\begin{theorem}\label{thm:CC-W1h}
  \CC-\mw{} is W[1]-hard with respect to the sum of misrepresentations~$\bound$.
\end{theorem}

\begin{theorem}\label{thm:M-W1h}
  \Monroe-\mw{} is W[1]-hard with respect to the sum of misrepresentations~$\bound$.
\end{theorem}

\section{Related Works}

The computational and algorithmic complexity study of representation voting rules for selecting multi-winners or a committee began two decades ago~\cite{procaccia2008complexity,potthoff1998proportional,lu2011budgeted}.
Potthoff and Brams~\cite{potthoff1998proportional} formulate the \Monroe-\mw{} and \CC-\mw{} problem via Integer Linear Programming, implying that both problems are contained in NP.
Lu and Boutilier~\cite{lu2011budgeted} showed NP-hardness for \CC-\mw{}.
Procaccia et al.~\cite{procaccia2008complexity} show that \CC-\mw{} and \Monroe-\mw{} are trivially contained in XP with respect to the number~$k$ of multi-winners since there are $\binom{m}{k}$ possible size-$k$ subsets of alternatives to be checked, where $m$ is the total number of alternatives. 

In terms of approximability results, Skowron et al. \cite{SFS2015} showed that it is NP-hard to approximate up to any constant factor for minimum misrepresentation under both rules.

 In terms of parameterized complexity, Betzler et al.~\cite{BetzlerMW2013} provide an extensive multi-variate complexity investigation of \CC-\mw{} and \Monroe-\mw{}, and their minimax variants.
The parameters that they consider are the number~$n$ of voters, the number~$m$ of alternatives,
the number~$k$ of multi-winners,
and the sum of misrepresentations. 
They show that both \CC-\mw{} and \Monroe-\mw{} are parameterized intractable with respect to~$k$, but are fixed-parameter tractable with respect to~$n$, or $m$, or $(k+\bound)$.
The only missing piece of the parameterized result puzzle is with respect to~$\bound$ alone,
i.e., whether a fixed-parameter algorithm with running time~$f(\bound)\cdot (n+m)^{O(1)}$ exists.
We settle this remaining question and show that such an algorithm does not exist unless FPT${}={}$W[1].

\section{Preliminaries}
For an integer~$t$, let $[t]$ denote the subset~$\{1,\cdots,t\}$.
We show all negative results via parameterized reductions from the W[1]-hard \Clique{} problem parameterized by the size of the clique~\cite{DF13}.
\decprob{\Clique}
{An undirected graph~$G=(U,E)$ and a non-negative integer~$h$.}
{Does $G$ admit a \myemph{clique} of size~$h$, i.e., a size-$h$ vertex subset~$U'\subseteq U$ such that the induced subgraph~$G[U']$ is complete?}

Since the reductions share a common construction, we first describe this construction as follows.
Briefly put, given an instance~$(G, h)$ of \Clique{}, we construct two gadgets, one for the vertices and the other for the edges.
No voters in the vertex gadget share the same most-preferred alternative with any voters in the edge gadget,
but they share the same second-ranked alternative if the corresponding vertex (of the voter in the vertex gadget) and edge (of the voter in the edge gadget) are incident to each other.
We add private blocker alternatives and append them to the voters' preference orders (starting from rank two) to ensure that for small sums of misrepresentations, each voter must be assigned an alternative with rank either zero or one.
Together with the number~$k$ of multi-winners, this ensures that $\binom{h}{2}$ groups of voters must be assigned to alternatives with rank one.
These voters correspond to $\binom{h}{2}$ edges~$E'$ and by the preferences, the alternatives that are assigned to the voters must correspond to the incident vertices of~$E'$.
The vertex gadget, together with~$k$ and the bound~$\bound$ (which is a function in~$h$) on the sum of misrepresentations, ensures that at most $h$ such incident vertices can be selected.
Altogether, this works only if $E'$ induces a clique of size~$h$.
\cref{constr:Base} gives the formal definition of the construction.
For the Monroe rule, we will extend \cref{constr:Base} to make sure that the desired assignment is proportional.

\begin{construction}\label{constr:Base}
Let $(G, h)$ denote an instance of \Clique\ with $G=(U,E)$ such that $U=\{u_1,\cdots,u_{\enn}\}$ and $E=\{e_1,\cdots, e_{\emm}\}$.

For each vertex~$u_i\in U$, we create a \myemph{vertex alternative}, called~$u_i$, a dummy vertex, called~$d_i$, and a set~$A_i$ of $\enn$ \myemph{blocker alternatives}. %
 Define $D\coloneqq \{d_i \mid i \in [\enn]\}$. %
 For each edge~$e_j\in E$,
 we create an \myemph{edge alternative}, called~$e_{j}$,
 and a set~$B_j$ of $\emm$ \myemph{blocker alternatives}. %
The role of the blocker alternatives is to ensure that only the vertex alternatives and edge alternatives will be selected in the target assignment.
 Note that we use $u_i$ (resp.\ $e_j$) for both the vertices and the vertex alternatives (resp.\ the edges and the edge alternatives) for better association; the meaning of these variables will be clear from the context.
 Together, the set~$C$ of alternatives is defined as
 \[C\coloneqq U\cup D\cup \bigcup\limits_{i\in [\enn]}A_i\cup E \cup \bigcup\limits_{j\in [\emm]}B_j.\]

 Before we describe the voters and their preferences, let us introduce some notation. 
 Given a subset of alternatives~$X$, the notation~``\myemph{$\seq{X}$}'' refers to an arbitrary but fixed ordering of the alternatives in~$X$, and the notation~``\myemph{$\cdots$}'' refers to an arbitrary but fixed ordering of the remaining alternatives not mentioned explicitly.
 Further, $\kk_1$ and $\kk_2$ are two numbers that only depend on $h$ (to be determined in the concrete reductions).
 Now, we create two groups of voters, one for the vertices and one for the edges, as follows:
 \begin{itemize}[--]
   \item For each vertex~$u_i\in U$, create $\kk_1$ \myemph{vertex voters}, called~$v^z_i$, $z\in [\kk_1]$, who each have the same preference order:
   $v_i^z\colon d_i\succ u_i \succ \seq{A_i} \succ \cdots$.
   \item For each edge~$e_j\in E$, let $u_i$ and $u_{i'}$ denote the two endpoints of~$e_j$.
   Then, create $2\kk_2$ \myemph{edge voters}, called~$w^z_j(u_i)$ and $w^z_j(u_{i'})$, $z\in [\kk_2]$, such that
   all voters~$w^z_j(u_{i})$ have the same preference order:
   $w^z_j(u_i)\colon e_j\succ u_i\succ \seq{B_j}\succ \cdots$,
   and  all voters~$w^z_j(u_{i'})$ have the same preference order:
   $w^z_j(u_{i'})\colon e_j\succ u_{i'}\succ \seq{B_j}\succ \cdots$.
 \end{itemize}
 In total, the voter set~$V$ is defined as
 \[V\coloneqq \{v_i^z\mid i\in [\enn], z\in [\kk_1]\}\cup \{w^z_j(u_i), w^z_j(u_{i'})\mid j \in [\emm] \text{ with } e_j=\{u_i,u_{i'}\}, z\in [\kk_2]\}.\]
\end{construction}

\section{Proof for \cref{thm:CC-W1h}}\label{sec:proof-CC-W1h}
In order to show W[1]-hardness, we provide a single parameterized reduction from \Clique.
Given an instance~$(G,h)$ of \Clique{} with $G=(U,E)$ and $U=\{u_1,\cdots, u_{\enn}\}$ and $E=\{e_1,\cdots, e_{\emm}\}$, let $\ppp=(C, V, \RR)$ denote the preference profile computed in \cref{constr:Base}.
Recall that for each vertex, the profile consists of $\enn$ vertex voters and for each edge it consists of $\emm$ edge voters.
Without loss of generality, assume that $\enn, \emm \gg h^3$ and $h \ge 3$ so that $\binom{h}{2} > 2$. %
To complete the construction, we set $\kk_1\coloneqq 1$ and $\kk_2\coloneqq h$.
Finally, we set the number~$k$ of desired alternatives and the misrepresentation bound~$\bound$ to be~$k\coloneqq \emm-\binom{h}{2} + \enn$ and \myemph{$\bound\coloneqq \kk_1\cdot h + \kk_2 \cdot 2\binom{h}{2}$}, respectively.
Clearly, $\bound < \min(\enn, \emm)$ and the whole construction can be done in polynomial time.
 
It remains to show the correctness. More precisely, we show the following.
\begin{enumerate}[(i)]
 \item\label{CC-forward} If the input graph~$G$ admits a size-$h$ clique, then $\ppp'$ admits an assignment~$\repre\colon V \to C$ such that
 $|\repre(V)| = k$ and $\misr(\repre, \ppp) \le \bound$.
 \item\label{CC-backward} If $\ppp'$ admits an  assignment~$\repre\colon V \to C$ such that
 $|\repre(V)| = k$ and $\misr(\repre, \ppp) \le \bound$,
 then the input graph~$G$ admits a size-$h$ clique. %
\end{enumerate}

\noindent For the forward direction~\eqref{CC-forward}, assume that $G$ admits a size-$h$ clique~$U'\subset U$. 
 Then, define an assignment~$\repre\colon V\to C$ as follows:
 \begin{itemize}[--]
   \item For each vertex voter~$v_i^z$, $i\in [\enn]$ and $z\in [\kk_1]$,
   let $\repre(v_i^z)\coloneqq u_i$ if $v_i\in V'$; otherwise let $\repre(v_i^z)\coloneqq d_i$.
   \item For each edge~$e_j \in E$, let $u_i$ and $u_{i'}$ denote the endpoints of~$e_j$. %
   Then, for each $z\in [\kk_2]$,
   let $\repre(e_j^z(u_i)) \coloneqq u_i$ and $\repre(e_j^z(u_{i'}))\coloneqq u_{i'}$ if $e_j\in E(G[U'])$ (i.e., both endpoints of $e_j$ are in the clique); otherwise
   $\repre(e_j^z(u_i)) = \repre(e_j^z(u_{i'})) \coloneqq e_j$. 
 \end{itemize}
 It is straightforward to verify that $\repre$ is the desired assignment.

 \noindent For the backward direction~\eqref{CC-backward}, assume
 that 
 $\ppp$ admits an assignment~$\repre$ with $|\repre(V)|=k$ and $\misr(\ppp, \repre)\le \bound$.
 We aim to show that $G$ admits a size-$h$ clique.
 Let $\repre$ denote an assignment with minimum sum of misrepresentations, implying that for any other assignment~$\repre'$ with $|\repre'(V)|=k$ it holds that $\misr(\ppp,\repre) \le \misr(\ppp,\repre')$.
 First, we analyze the alternatives in $\repre(V)$ and observe the following.
 \begin{claim}\label{claim:relevant-alts}
   For each voter~$v \in V$ it holds that $\rank_v(\repre(v)) \le 1$.
 \end{claim}
 \begin{proof}\renewcommand{\qedsymbol}{(of
     \cref{claim:relevant-alts})~$\diamond$}
   Suppose, towards a contradiction, that there exists a voter~$v\in V$ such that $\rank_v(\repre(v)) > 1$.
   Then, such a voter is either a vertex voter or an edge voter.
   If this is a vertex voter, say~$v_i^z$, then there must be an index~$x\in [\enn]$ such that neither $d_x$ nor $u_x$ is contained in~$\repre(V)$.
   If $u_i, d_i\notin \repre(V)$, then set $x$ is such a desired index. 
   Otherwise, we have that $u_i\in\repre(V)$ or $d_i\in \repre(V)$.
   If more than two voters are assigned to~$\repre(v_i^z)$,
   then it is straightforward to verify that reassigning $d_i$ (resp.\ $u_i$) to $v_i^z$ (keeping the remaining assignment intake) decrease the sum of misrepresentations.
   This implies that only $v_i^z$ is assigned to $\repre(v_i^z)$.
   Then, we claim that there exists another assignment with smaller sum of misrepresentations than~$\repre$,
   again contradicting our choice of $\repre$.
   We distinguish between two cases.
   If $E\nsubseteq \repre(V)$, meaning that there exists an edge voter~$e_j\in E\setminus \repre(V)$, then let $w\in \{w_j^{z'}(u_x), w_j^{z'}(u_y) \mid e_j =\{u_x,u_{y}\}, z'\in [\kk_2]\}$. 
   It is straightforward to verify that the following reassignment~$\repre'$ satisfies $|\repre'(V)|=k$ and has smaller sum of misrepresentations than~$\repre$; recall that $\kk_2=h>1$ and no voter other than $v_i^z$ is assigned $\repre(v_i^z)$: 
   For all $v\in V$, define
   \begin{align*}
     \repre(v) \coloneqq 
     \begin{cases}
       \repre'(v) = d_i, &  \text{ if } v = v_i^z,\\
       \repre'(v) = e_j, & \text{ if } v = w, \\
       \repre'(v), & \text{ otherwise.}
     \end{cases}
   \end{align*}

   If $E\subseteq \repre(V)$, then $|(D\cup U)\cap \repre(V)| \le n-\binom{h}{2} < n$.
   This means that there exists an index~$x \in [\enn]\setminus \{i\}$ with $d_x, u_x\notin \repre(V)$, as desired.

   We just showed that $x\in [\enn]$ is an index such that neither $d_x$ nor $u_x$ is contained in $\repre(V)$.
   Since $\misr(\ppp,\repre) \le \bound < \min(\emm, \enn)$, for each $z\in [\kk_1]$,
   vertex voter~$v_x^z$ must be assigned
   to an associated blocker alternative~$a\in A_x$.
   Now, observe that $|A_{i'}|= \enn > \bound$ holds for all~$i' \in [\enn]$
   and $|B_j|=\emm > \bound$ holds for all~$j\in [\emm]$,
   so no voter other than that from~$V_x\coloneqq \{v_x^{z}\mid z\in [\enn]\}$ is assigned to~$a$ as otherwise the sum of misrepresentations exceeds~$\bound$.
   Thus, we can replace~$a$ with $d_x$ in the assignment of every voter in~$V_x$ so to decrease the sum of misrepresentations, a contradiction.

   The reasoning for the case when such a voter is an edge voter works analogously.  %
 \end{proof}

 By the preferences of the vertex voters, \cref{claim:relevant-alts} immediately implies the following,
 \begin{align}
   \label{eq:relev-alts}
   \repre(V)\subseteq U \cup D \cup E.
 \end{align}

 Now, we claim that $\repre(V)$ corresponds to a size-$h$ clique.
 To this end, define $E'\coloneqq E\setminus \repre(V)$,
 $D'\coloneqq D\cap \repre(V)$, and $U'\coloneqq \{u_i\in U \cap \repre(V) \mid d_i\notin D'\}$.
 Clearly, $|U'|+|D'|\le \enn$.
 By \cref{claim:relevant-alts} and by the preferences of the vertex voters, we further infer that $|U'|+|D'| = \enn$.
 Together \eqref{eq:relev-alts}, it follows that
 \begin{align}
   \label{eq:size-rest}   |E \cap \repre(V)| = \emm-\binom{h}{2}.
 \end{align}
 This further implies that
 \begin{align}
   |E'| = \binom{h}{2}.\label{eq:size-E'}
 \end{align}

 In order to show that $G=(V,E)$ admits a size-$h$ clique,
 it suffices to show that $|E'| \ge \binom{h}{2}$ and $|\{u_i,u_{i'}\mid \{u_i,u_{i'}\} \in E'\}| \le h$.
 Let $U^*\coloneqq \{u_i,u_{i'}\mid \{u_i,u_{i'}\} \in E'\}$, by \eqref{eq:size-E'}, it remains to show that
 $|U^*| \le h$.
 Suppose, towards a contradiction, that $|U^*| > h$.
 By the misrepresentation bound~$\bound$,
 it follows that
 \begin{align*}
  \kk_1\cdot h_1 + \kk_2\cdot 2 \binom{h}{2}\ge \misr(\repre,\RR) \ge (\enn-|D'|)\cdot \kk_1 +  2|E'|\cdot \kk_2 \stackrel{\eqref{eq:size-E'}}{\ge} (\enn-|D'|)\cdot \kk_1 + 2\binom{h}{2}\cdot \kk_2.
 \end{align*}
 The second inequality holds since there are exactly $\kk_1 \cdot (\enn-|D'|)$ vertex voters and
  $2\kk_1\cdot |E'|$ edge voters that are each assigned an alternative with positive rank.
 Hence, $|D'|\ge \enn -h$.
 Further, by \cref{claim:relevant-alts}, it follows that $U^*\subseteq \repre(V)$.
 Then, by the bound on~$k$, 
 it follows that
 \begin{align*}
   \emm-\binom{h}{2}+\enn = |\repre(V)| \ge \emm - |E'| + |D'| + |U^*| > \emm - |E'| + (\enn-h) + h = \emm - |E'| + \enn.
 \end{align*}
 In other words, we have that $|E'| > \binom{h}{2}$.
 This yields a contradiction to the sum of misrepresentations:
 \begin{align*}
   \misr(\repre,\RR) \ge 2|E'|\cdot \kk_2 \ge 2(\binom{h}{2}+1)\cdot \kk_2 > 2\binom{h}{2}\cdot \kk_2 + 2\kk_2 > 2\binom{h}{2}\cdot \kk_2 + \kk_1\cdot h = \bound.
 \end{align*}
 Note that the last but one inequality holds since $\kk_1 = 1$ and $\kk_2 = h$.
 This completes the correctness proof.
 Since the bound~$\bound$ is a function depending only on~$h$,
 we indeed obtain a parameterized reduction and show that \CC-\mw{} is W[1]-hard with respect to the bound~$\bound$ on the sum of misrepresentations.

\section{Proof for \cref{thm:M-W1h}}\label{sec:proof-M-W1h}
To show hardness for the Monroe rule, we extend the instance constructed in \cref{sec:proof-CC-W1h}.
Given an instance~$(G,h)$ of \Clique{} with $G=(U,E)$ and $U=\{u_1,\cdots, u_{\enn}\}$ and $E=\{e_1,\cdots, e_{\emm}\}$, let $C$, $V$, and $\RR$ denote the set of alternatives, the set of voters, and the collection of preference orders computed in \cref{constr:Base}, respectively.
Before we detail the reduction, we analyze the number of voters represented by an alternative in a desired assignment in the original reduction.
Recall that in a desired assignment (for the Chamberlin and Courant rule),
each assigned vertex alternative represents $\kk_1$ vertex voters and $\kk_2\cdot (h-1)$ edge voters, each assigned dummy alternative represents $\kk_1$ vertex voter, while each assigned edge alternative represents~$2\kk_2$ edge voters.
Now, to adapt the reduction to work for the Monroe rule, we aim to add additional desired alternatives and new voters to the profile to enforce that each assigned alternative must represent the same number of voters.
This number is set to~$L\coloneqq \kk_1 + \kk_2\cdot (h-1)$, which is exactly the number of voters represented by an assigned vertex alternative in the original reduction.

Specifically, we extend the set of alternatives by adding some more blocker alternatives and some new voters which ensure that each desired alternative represents the same portion of voters.
For each vertex~$u_i\in U$, add to~$C$ a set~$\hat{A}_i$ of $\enn$ blocker alternatives.
For each edge~$e_j\in E$, add to~$C$ a set~$\hat{B}_j$ of $\emm$ blocker alternatives.
Add to~$C$ two sets of alternatives, $X$ and $Y$, with $X=\{x_z\mid z \in [\kk_2]\}$ and $Y=\{y_z\mid z \in [L-2\kk_2]\}$; note that $2|X|+|Y| = L$.
Finally, for each~$c\in X\cup Y$, add a set~$F_c$ of $\enn+\emm$ blocker alternatives.
In total, the set of alternatives is

\[C\coloneqq U\cup D\cup \bigcup\limits_{i\in [\enn]}(U_i\cup \hat{A}_i) \cup E \cup \bigcup\limits_{j\in [\emm]}(B_j\cup \hat{B}_j) \cup X\cup Y \cup \bigcup\limits_{c\in X\cup Y} F_c.\]
The size of the assigned alternatives is set to $k\coloneqq \emm-\binom{h}{2}+\enn+|X|+|Y| = \emm+\enn-\binom{h}{2}+\kk_1+\kk_2(h-2)$.

To the voter set and the collection of preference orders we add the following new voters,
where we use the notations $\seq{\cdot}$ and $\cdots$ analogously as in \cref{sec:proof-CC-W1h}.
\begin{itemize}[--]
  \item For each~$i\in [\enn]$ and each~$z\in [|X|]$, add to~$V$ a set~$\hat{V}^z_{i}$ of $h-1$ voters, all with the same preference order:
  \begin{align*}
    d_i\succ x_z \succ \seq{\hat{A}_z} \succ \cdots.
  \end{align*}
  Notice that there are exactly $|X|\cdot (h-1) = \kk_2\cdot (h-1)$ voters who each rank~$d_i$ ($i\in [\enn]$) at the first position.
  \item For each~$j\in [\emm]$ and each~$z\in [|Y|]$, add to~$V$ a voter, called~$\hat{w}^z_j$, with preference order:
  \begin{align*}
    \hat{w}_j^z\colon e_j\succ y_z \succ \seq{\hat{B}_z}\succ \cdots.
  \end{align*}
  Define $\hat{W}_j \coloneqq \{\hat{w}^z_j\mid z\in [|Y|]\}$.
  \item For each~$z\in [|X|]$, add a set~$\hat{V}^z_0$ of $L-h\cdot (h-1)$ special voters, all with the same preference order: 
  \begin{align*}
     x_z \succ \seq{F_{x_z}}\succ \cdots.
  \end{align*}
  \item For each~$z\in [|Y|]$, add a set~$\hat{W}_0^z$ of~$L-\binom{h}{2}$ special voters, all with the same preference order:
  \begin{align*}
    \hat{w}^z_0\colon y_z \succ \seq{F_{y_z}}\succ \cdots.
  \end{align*}
\end{itemize}
In total, the voter set is
\begin{align*}
  V\coloneqq & \{v_i^z\mid i\in [\enn], z\in [\kk_1]\}\cup \{w^z_j(u_i), w^z_j(u_{i'})\mid j \in [\emm] \text{ with } e_j=\{u_i,u_{i'}\}, z\in [\kk_2]\} \cup\\
  & \quad \bigcup\limits_{i\in [\enn],z\in [|X|]}\hat{V}^z_i \cup \bigcup\limits_{j\in [\emm]}\hat{W}_j \cup \bigcup\limits_{z\in [|X|]}\hat{V}_0^z \cup \bigcup\limits_{z\in [|Y|]}\hat{W}^z_{0}.
\end{align*}
This gives
\begin{align*}
  |V| & = \kk_1\cdot \enn+2\kk_2\cdot \emm+ |X|\cdot (h-1)\cdot \enn + |Y| \cdot \emm + (L-h\cdot (h-1)) \cdot |X| + (L-\binom{h}{2}) \cdot |Y|\\
  & = \kk_1 \cdot \enn + (2\kk_2+|Y|)\cdot \emm + |X|\cdot (h-1)\cdot \enn + L\cdot (|X|+|Y|) - \binom{h}{2}\cdot (2|X|+|Y|)\\
      & = L \cdot (\emm + \enn -\binom{h}{2} + |X|+|Y|) = L \cdot k,
\end{align*}
where the last inequality holds since $2|X|+|Y|=L$ and $|X|=\kk_2$.
This implies that each assigned alternative must represent exactly $L$ voters.
Let $\ppp'$ denote the extended preference profile.
The bound on the sum of misrepresentations is set to
$\bound \coloneqq \kk_1\cdot h+2\binom{h}{2}\cdot \kk_2 + \binom{h}{2}\cdot L$.
Finally, we set $\kk_1=1$ and $\kk_2=h$, which makes $\bound$ upper-bounded by $O(h^4)$ since $L=\kk_1+(h-1)\cdot \kk_2$.

It remains to show the correctness. More precisely, we show the following.
\begin{itemize}[--]
 \item If the input graph~$G$ admits a size-$h$ clique, then $\ppp'$ admits an assignment~$\repre\colon V \to C$ such that
 $|\repre(V)| = k$ and $\misr(\repre, \ppp') \le \bound$.
 \item  If $\ppp'$ admits an  assignment~$\repre\colon V \to C$ such that
 $|\repre(V)| = k$ and $\misr(\repre, \ppp) \le \bound$,
 then the input graph~$G$ admits a size-$h$ clique. %
\end{itemize}

The forward direction is again straightforward.
Assume that $G$ admits a size-$h$ clique~$U'\subset U$. 
Then, define an assignment~$\repre\colon V\to C$ as follows:
 \begin{itemize}[--]
   \item For each $i\in [\enn]$, if
   $u_i\in U\setminus U'$, then for each voter~$v\in \{v_i^r\mid r\in [\kk_1]\}\cup \bigcup\limits_{z\in [|X|]}\hat{V}_i^z$,
   assign $\repre(v)\coloneqq d_i$;
   otherwise, for each voter~$v\in \{v_i^r\mid r\in [\kk_1]\}$,
   assign $\repre(v)\coloneqq u_i$,
   and for each $z\in [|X|]$ and each voter~$\hat{v}\in \hat{V}_i^z$,
   assign~$\repre(\hat{v}) \coloneqq  x_z$.
   Note that in this way, we have that for each $i\in [\enn]$, either $|\repre^{-1}(d_i)|=0$ or $|\repre^{-1}(d_i)|=\kk_1+|X|\cdot (h-1) = L$ holds.

   \item For each edge~$e_j \in E$, let $u_i$ and $u_{i'}$ denote the endpoints of~$e_j$.
   Then, for each $z\in [\kk_2]$,
   if $e_j\in E(G[U'])$  (i.e., both endpoints of $e_j$ are in the clique), 
   then assign $\repre(e_j^z(u_i)) \coloneqq u_i$ and $\repre(e_j^z(u_{i'}))\coloneqq u_{i'}$ if $e_j\in E(G[U'])$,
   and for each $\hat{w}^z_j \in \hat{W}_j$,
   assign $\repre(\hat{w}^z_j) \coloneqq  y_z$;
   otherwise, assign   $\repre(e_j^z(u_i)) = \repre(e_j^z(u_{i'})) \coloneqq e_j$,
   and for each voter~$\hat{w}\in \hat{W}_j$
   assign $\repre(\hat{w}) \coloneqq e_j$.
   Note that for each $j \in [\emm]$,  either $|\repre^{-1}(e_j)|=0$ or $|\repre^{-1}(e_j)|=2\kk_2 + |Y| = L$ holds.
   \item For each $z\in [|X|]$ and each voter~$\hat{v}\in \hat{V}_0^z$, assign $\repre(v)\coloneqq x_z$, and for each $z\in [|Y|]$ and each voter~$\hat{w}\in \hat{W}_0^z$, assign $\repre(w)\coloneqq y_z$.
 \end{itemize}
 One can verify that $|\repre(V)| = k$ and that for each assigned alternative~$c\in \repre(V)$, it holds that
 $|\repre^{-1}(c)| = L$.
 Further, the sum of misrepresentations is exactly~$\bound$.
 
 For the backward direction, assume that  $\ppp'$ admits a desired proportional assignment.
 Let $\repre$ denote such an assignment with $|\repre(V)| = k$ such that for each assigned alternative~$c\in \repre(V)$, it holds that $|\repre^{-1}(c)| = L$,
 and for each other proportional assignment~$\repre'$ it holds that $\misr(\repre,\RR) \le \misr(\repre',\RR)$.

 By the size of the blocker alternatives and by the preferences of the voters, we infer the following.
 \begin{claim}\label{claim:relevant-alts2}
   $X\cup Y \subseteq \repre(V)$ holds and for each $v \in V$ it holds that $\rank_v(\repre(v)) \le 1$.
 \end{claim}

 \begin{proof}
   \renewcommand{\qedsymbol}{(of
     \cref{claim:relevant-alts2})~$\diamond$}
   All two relations are straightforward.
   To see why the first relation holds, we observe that if this were not the case,
   then there must exist a voter~$v\in \hat{V}_0^z \cup \hat{W}_0^{r}$, $z\in [|X|]$, $r\in [|Y|]$ which is represented by some blocker alternative, say $c$, since $\enn, \emm > \bound$.
   However, less than $L$ voters rank this alternative in the first $\bound+1$ positions,
   which implies that no proportional assignment can assign some voter to this alternative~$c$ without exceeding the sum of misrepresentations.
   
   Now that we have shown the first relation, the second relation follows by a similar reasoning as the one given for \cref{claim:relevant-alts}.
 \end{proof}

 By the preferences of the vertex voters, the above immediately implies the following,
 \begin{align}
   \label{eq:relev-alts2}
   \repre(V)\subseteq U \cup D \cup E \text{ with } |\repre(V)\cap (U\cup D\cup E)| = m-\binom{h}{2}+n.
 \end{align}

 As already discussed, by the size of $k$ and the number of voters,
 each assigned alternative must represent exactly $L$ voters.
 By \cref{claim:relevant-alts2}, each alternative~$x_z\in X$ must represent $h\cdot (h-1)$ voters from $\bigcup\limits_{i\in [\enn]}\hat{V}_i^z$.
 Similarly, each alternative~$y_z\in Y$ must represent $\binom{h}{2}$ voters from $\{\hat{w}_j^z\mid j \in [\emm]\}$.
 This yields a sum of misrepresentation of $|X|\cdot h\cdot (h-1)+|Y|\cdot \binom{h}{2}= L\cdot \binom{h}{2}$.
 By the definition of $\bound$, it follows that the remaining alternatives can yield a sum of misrepresentations of $h \cdot \kk_1 + 2\binom{h}{2}\cdot \kk_2$.
 Having established this, the remaining proof works analogously as the one given in \cref{thm:CC-W1h}.

\bibliographystyle{abbrvnat}

\bibliography{bib}

\end{document}